\documentclass[twocolumn,shownopacs,floatfix,prl]{revtex4}

\usepackage{graphicx}
\usepackage{array}
\usepackage[caption=false]{subfig}
\usepackage{amsmath,amssymb,amsthm}

\newtheorem{thm}{Theorem}
\newtheorem{lemma}{Lemma}

\newcommand{\ket}[1]{\left|#1\right\rangle}

\newcommand{\Z}{\mathbb{Z}}
\newcommand{\N}{\mathbb{N}}
\newcommand{\C}{\mathbb{C}}
\newcommand{\R}{\mathbb{R}}

\newcommand{\AI}{A{\rm I }}
\newcommand{\AII}{A{\rm I\!I }}
\newcommand{\AIII}{A{\rm I\!I\!I }}
\newcommand{\CI}{C{\rm I }}
\newcommand{\CII}{C{\rm I\!I }}
\newcommand{\BDI}{BD{\rm I }}
\newcommand{\DIII}{D{\rm I\!I\!I }}
\begin{document}

\title{Homotopy Theory of Strong and Weak Topological Insulators}

\author{R. Kennedy, C. Guggenheim}
\affiliation{Institute for Theoretical Physics, University of Cologne}

\begin{abstract}
We use homotopy theory to extend the notion of strong and weak topological insulators to the non-stable regime (low numbers of occupied/empty energy bands). We show that for strong topological insulators in $d$ spatial dimensions to be ``truly $d$-dimensional'', i.e. not realizable by stacking lower-dimensional insulators, a more restrictive definition of ``strong'' is required. However, this does not exclude weak topological insulators from being ``truly $d$-dimensional'', which we demonstrate by an example. Additionally, we prove some useful technical results, including the homotopy theoretic derivation of the factorization of invariants over the torus into invariants over spheres in the stable regime, as well as the rigorous justification of the parameter space replacements $T^d\to S^d$ and $T^{d_k}\times S^{d_x}\to S^{d_k+d_x}$ used widely in the current literature.
\end{abstract}

\maketitle

\section{Introduction}

In recent years there has been considerable interest in topological phases of band insulators, fuelled by their theoretical prediction \cite{bernevig,fukane,zhang} and subsequent experimental realization \cite{molenkamp,hasan,hasan2} in two and three dimensional time-reversal invariant systems. On the theoretical side the obvious question arose: Under what circumstances (dimensions, symmetries, ...) do topological phases occur? For the case of free fermions, a partial answer was given in the seminal paper \cite{kitaev} using $K$-theory, the result of which is displayed in Table~\ref{periodic} (the "periodic table of topological insulators and superconductors"). It is only a partial answer, since it has a major limitation: It is assumed that the number of occupied as well as empty energy bands is large enough for $K$-theory to apply. In this \textit{stable regime}, it can be shown using $K$-theory~\cite{kitaev} that invariants defined for a torus $T^d$ as momentum space factorize into a product of $d \choose 
l$ independent invariants over $S^l$, where $l$ runs from $0$ to $d$ and Table~\ref{periodic} shows the answer for each $l$. We give an alternative, homotopy theoretic derivation of this result. The factorization of the invariants in the stable regime leads to a natural definition of \textit{strong} topological insulators as those insulators that have a non-trivial invariant in the factor with domain $S^d$, while the remaining insulators are dubbed \textit{weak} topological insulators~\cite{fukanemele}. This distinction, which we refer to as definition (i), becomes problematic in the non-stable regime. In this paper, we propose an alternative definition of strong/weak topological insulators (definition (ii)), which defines a topological insulator to be strong only if it has a non-trivial invariant over $S^d$ \textit{and all other invariants in the product are trivial}, while all remaining insulators are called weak. The additional restriction is necessary in order to guarantee strong topological insulators 
to be ``truly $d$-dimensional''. We demonstrate this by providing an example in the non-stable regime which is strong according to definition (i), but can be realized by stacking lower-dimensional topological insulators. The latter cannot happen for definition (ii).

After formalising the notion of stacking insulators, we further demonstrate through an example in the stable regime that for both definitions, there may be weak topological insulators that cannot be realized through stacking. In other words, there may be also be ``truly $d$-dimensional'' weak topological insulators.

In order for definition (ii) to be well-defined, we prove that distinct phases over a sphere as momentum space always remain distinct over the torus with the same dimension as momentum space. By a straightforward generalization, we show that they also always remain distinct over any product of sphere and torus with the same (total) dimension, which provides a rigorous justification for the replacement $S^{d_x}\times T^{d_k}\to S^{d_x+d_k}$ in~\cite{teokane}, where invariants for topological insulators in $d_k$ dimensions with a defect of codimension $d_x+1$ are calculated.

In this paper, we use the natural notion of homotopy (also known as adiabatic or continuous deformation) as an equivalence relation between insulator ground states. This generalizes two definitions of topological insulators in the current literature: The first one is based on $K$-theory and the generalization is the extension to the non-stable regime. The second definition, as adopted in \cite{hughes}, defines an insulator to have non-trivial topology if there is no adiabatic deformation to the atomic limit, where fermions are localized at lattice points. While this second definition already uses the notion of homotopy (= adiabatic deformation), it only distinguishes non-trivial from trivial and is generalized here by additionally distinguishing between different non-trivial insulators.

\section{Setting and statement of results}

A translation invariant free fermion ground state of an insulator in one of the two ``complex'' symmetry classes (upper two rows in Table~\ref{periodic}) is described by a map
\begin{align}
\psi:T^d\to C_s,
\end{align}
where $T^d$ is the Brillouin zone torus and $C_s$ is either a Grassmannian Gr$_m(\mathbb{C}^n)$ for even~$s$ (class A) describing an $n$-band model with $m$ filled bands, or a unitary group $U_n$ for odd $s$ describing a $2n$-band model with $n$ filled bands in the chiral class \AIII. A ground state in one of the eight ``real'' classes (eight lower rows in Table~\ref{periodic}) satisfies the additional requirement
\begin{align}
\tau_s(\psi(k))=\psi(-k),\label{equivariance}
\end{align}
restricting the image at time reversal invariant momenta to the space $R_s$ (see Table~\ref{periodic}). An equivalent formulation of condition~\eqref{equivariance} is to say that $\psi$ is \textit{equivariant} with respect to the $\Z_2$-actions given by $k\mapsto-k$ on $T^d$ and $\tau_s$ on $C_s$. A detailed description of all $C_s$, $R_s$ and $\tau_s$ can be found in~\cite{kz}.

\begin{table}
\centering
\begin{ruledtabular}
\begin{tabular}{cc|c|c|cccccccc}
\multicolumn{2}{c|}{symmetry} & target & fixed point set & \multicolumn{4}{c}{$[S^d,C_s]^*_{\Z_2}$} \\
$s$& label & $C_s$ & $R_s$ &
$0$ & $1$ & $2$ & $3$\\
\hline
even & A 	& $\cup_{m=0}^n\text{Gr}_m(\C^n)$ 		& $\cup_{m=0}^n\text{Gr}_m(\C^n)$ 			& $\Z_{n+1}$ & $0$ & $\Z$ & $0$\\
odd & \AIII 	& $U_n$ 						& $U_n$							& $0$ & $\Z$ & $0$ & $\Z$\\
\hline
0 & D 		& $\cup_{m=0}^{2n}\text{Gr}_m(\C^{2n})$	& $O_{2n}/U_{n}$ 						& $\Z_2$ & $\Z_2$ & $\Z$ & $0$\\
1 & \DIII 	& $U_{2n}$ 						& $U_{2n}/Sp_{2n}$ 					& $0$ & $\Z_2$ & $\Z_2$ & $\Z$ \\
2 & \AII 	& $\cup_{m=0}^n\text{Gr}_{2m}(\C^{2n})$ 	& $\cup_{m=0}^n\text{Gr}_m(\mathbb{H}^n)$	& $\Z_{n+1}$ & $0$ & $\Z_2$ & $\Z_2$\\
3 & \CII 	& $U_{2n}$ 						&  $Sp_{2n}$ 						& $0$ & $\Z$ & $0$ & $\Z_2$\\
4 & C 		& $\cup_{m=0}^{2n}\text{Gr}_m(\C^{2n})$	& $Sp_{2n}/U_n$ 						& $0$ & $0$ & $\Z$ & $0$\\
5 & \CI 	& $U_n$ 						& $U_n/O_n$						& $0$ & $0$ & $0$ & $\Z$\\
6 & \AI 	& $\cup_{m=0}^n\text{Gr}_m(\C^n)$ 		& $\cup_{m=0}^n\text{Gr}_m(\R^n)$ 			& $\Z_{n+1}$ & $0$ & $0$ & $0$\\
7 & \BDI 	& $U_n$ 						& $O_{n}$ 							& $\Z_2$ & $\Z$ & $0$ & $0$
\end{tabular}
\end{ruledtabular}
\caption{This table (the "periodic table of topological insulators and superconductors"~\cite{kitaev}) lists the sets $[S^d,C_s]^*_{\Z_2}$ for dimensions $0\le d\le3$ for the eight real and two complex symmetry classes indexed by $s$ mod $8$. In the complex case, the involution on $S^d$ and $C_s$ is trivial, while in the real case, it is non-trivial with fixed point sets $S^0$ and $R_s$ respectively.}
\label{periodic}
\end{table}

A topological phase in this setting is an equivalence class of ground state maps $\psi$, denoted $[\psi]$. In this paper, we use the equivalence relation of being homotopic, i.e. having a continuous interpolation respecting the additional equivariance condition~\eqref{equivariance} in the real cases.  The set of all topological phases in $d$ dimensions will be denoted $[T^d,C_s]$ for the complex classes and $[T^d,C_s]_{\Z_2}$ for the real ones. We will mainly use the latter since it includes the former as the special case with trivial $\Z_2$-action on $T^d$ and $C_s$, which is automatically equivariant.

As outlined in the introduction, this definition of a topological phase refines two definitions given in the current literature. One of them views an insulator ground state as a vector bundle of occupied eigenstates (of a Hamiltonian) over the Brillouin zone and defines equivalence classes through the notion of isomorphism of vector bundles~\cite{denittis1,denittis2}. These bundles may be constructed using pullback under $\psi$ (the \textit{classifying map}) of the tautological bundle over the Grassmannian. If there is no isomorphism between the vector bundles associated to two ground state maps $\psi_0$ and $\psi_1$, then there is no homotopy between them, i.e. $[\psi_0]\ne[\psi_1]$, but the converse is only true for a large number of empty bands (large~$n-m$). The equivalence relation of isomorphisms of vector bundles is relaxed further when going to $K$-theory, where the equivalence relation called \textit{stable equivalence} only requires isomorphism of vector bundles up to direct sums with arbitrary trivial bundles~\cite{kitaev,stone}. Here, two vector bundles that represent different stable equivalence classes in $K$-theory are in particular not isomorphic, but the converse is only true for large bundle dimensions (large number $m$ of occupied bands). Therefore, the notion of homotopy classes refines that of isomorphism classes of vector bundles, which in turn refines that of stable equivalence in $K$-theory. Note that the intermediate step of considering isomorphism classes of vector bundles is limited to classes A, \AII and \AI where we have two parameters $m$ and $n$ (see Table~\ref{periodic}).

The other definition of topological phases~\cite{hughes} uses the atomic limit as a reference ground state, which corresponds to a constant map $\psi$. It defines an insulator ground state to have non-trivial topology if there exists no adiabatic deformation to the atomic limit. This translates to no homotopy existing to the constant map. Using homotopy as an equivalence relation also refines this definition since it additionally distinguishes between different non-trivial states.

We note here that the notion of homotopy is natural in that it is the direct mathematical formalization of the physical concept of "adiabatically connecting": Two ground states can be adiabatically connected \textit{if and only if} there is a homotopy between them. This implies that two insulators in different topological phases cannot be adiabatically connected without a topological phase transition (closing of the energy gap). While, therefore, the situation is clear for translation invariant, infinitely extended systems, the challenging task remains of establishing the bulk-boundary correspondence, which states that the boundary of a topological insulator is gapless. The bulk-boundary correspondence has been addressed primarily in the stable regime~\cite{gurarie,graf}, but numerical results indicate that it also holds in the non-stable regime~\cite{hopf}. However, spatial symmetries (symmetries that do not commute with translations) may violate it in either regime (see~\cite{hughes}, section 
IV~B~4).

Based on the refined definition of equivalence, we propose to define strong topological insulators as the non-trivial classes in the subset
\begin{align}
[S^d,Y]_{\Z_2}\subset[T^d,Y]_{\Z_2},\label{inclusion1}
\end{align}
while weak topological insulators are defined as the complement. We will refer to this new definition as definition~(ii). In appendix B, we prove that the above inclusion is indeed well-defined. Using similar arguments, we show that there is also an inclusion
\begin{align}
[S^{d_x+d_k},Y]_{\Z_2}\subset[S^{d_x}\times T^{d_k},Y]_{\Z_2},\label{inclusion2}
\end{align}
a useful technical result for determining topological phases in the presence of (single) defects with codimension $d_x+1$, which is used (without proof for the non-stable regime) in~\cite{teokane}.

The result~\eqref{inclusion1} and the corresponding definition (ii) is to be contrasted with the usual definition (i) (see~\cite{fukanemele}) in the stable regime. In terms of homotopy theory, the latter is based on the decomposition of $[T^d,C_s]_{\Z_2}$ into a product containing ${d \choose l}$ factors of each set $[S^l,C_s]^*_{\Z_2}$ with $l=0,\dots,d$:
\begin{align}
[T^d,C_s]_{\Z_2}\simeq\prod_{l=0}^d\Big([S^l,C_s]^*_{\Z_2}\Big)^{d \choose l},\label{decomposition}
\end{align}
where $[X,Y]^*_{\Z_2}$ denotes the set of homotopy classes of maps $X\to Y$ mapping a base point $x_0$ in the fixed point set $X^{\Z_2}$ to a base point $y_0\in Y^{\Z_2}$ with all homotopies respecting this property. This formula holds for all symmetry classes $s$ with a slight modification for classes A, \AI and \AII, where we replace $C_s$ by its connected components $(C_s)_0$ containing the base point (i.e. we fix a number of occupied bands) and omit the factor with $l=0$ on the right hand side. The decomposition above can be obtained via $K$-theory~\cite{kitaev}, but we give an independent homotopy theoretic proof in appendix A.

As an example, consider three-dimensional insulators in class \AII~\cite{fukanemele}. In that case,
\begin{align}
[T^3,(C_2)_0]_{\Z_2}\simeq\Z_2\times(\Z_2\times\Z_2\times\Z_2),
\end{align}
since $[S^3,C_2]^*_{\Z_2}=[S^2,C_2]^*_{\Z_2}=\Z_2$ and $[S^1,C_2]^*_{\Z_2}=0$ (we use ``0'' here to denote the set with only one element). The (original) definition (i) of a strong topological insulator~\cite{fukanemele} is the subset of $[T^d,C_s]_{\Z_2}$ which is non-trivial in the factor $[S^d,C_s]^*_{\Z_2}$. In the example, this would give a set of \textit{eight} strong topological insulators, whereas our definition (ii) only gives \textit{one} strong topological insulator. In the language of~\cite{fukanemele}, there is a strong invariant $\nu_0$ and three weak invariants $\nu_i$, $i=1,2,3$, all taking values $0$ or $1$ so that every phase is described by the tuple $(\nu_0;\nu_1,\nu_2,\nu_3)$. A strong insulator in~\cite{fukanemele} is any element of the form $(1;\nu_1,\nu_2,\nu_3)$, giving 8 possibilities, while with our definition only the element $(1;0,0,0)$ is strong. In both definitions, \textit{weak} topological insulators are the complementary set respectively.

An example which shows that the factorization in~\eqref{decomposition} may not hold in the non-stable regime is given by the Hopf insulator~\cite{hopf}. This is a three-dimensional insulator with one occupied and one empty band in class A. The topological phases of this model have been computed in~\cite{hopfcalc}:
\begin{align}
\begin{split}
[T^3,\text{Gr}_1(\mathbb{C}^2)]=\{(&n_0;n_1,n_2,n_3)\mid n_1,n_2,n_3\in\Z;\\
&n_0\in\Z\text{ for }n_1=n_2=n_3=0\text{ and}\\
&n_0\in\Z_{2\cdot\gcd(n_1,n_2,n_3)}\text{ otherwise}\}\label{hopf}
\end{split}
\end{align}
This example demonstrates that the invariants may not be independent of each other as in the stable regime. Only for $n_1=n_2=n_3=0$ do we find $n_0\in\Z=[S^3,\text{Gr}_1(\mathbb{C}^2)]\subset[T^3,\text{Gr}_1(\mathbb{C}^2)]$ and definition (ii) only considers the non-trivial elements in this subset to be strong topological insulators.

Another example of a strong invariant "breaking down" when lower dimensional invariants become non-trivial will be presented in the next section. Crucially, this example will demonstrate that strong topological insulators according to definition (i) may be realized by stacking lower-dimensional systems. Hence, only definition (ii) has the property that all strong topological insulators in $d$ dimensions are "truly $d$-dimensional" (we use this phrase synonymously for the property of not being realizable by stacking lower-dimensional systems).

\section{Stacked insulators}

Before we introduce the example that motivates the use of definition (ii) over (i), we introduce a formalism to deal with the possibility of stacking insulators into higher dimensions. 

Let the translation-invariant Hamiltonian $\hat H$ of an $n$-band model act on the Hilbert space \mbox{$l^2(\Z^d)\otimes\mathbb{C}^n$} of a $d$-dimensional lattice. In a basis $\{\ket{x,\alpha}\}$, with $x\in\Z^d$ and $\alpha=1,\dots,n$,
\begin{align}
\hat H\ket{x,\alpha}:=\sum_{x',\beta}h_{\alpha\beta}(x')\ket{x+x',\beta},\label{hamiltonian}
\end{align}
where $h_{\alpha\beta}(x')=\overline{h_{\beta\alpha}(-x')}$ to ensure hermiticity. 

In the following, we will fix the given basis and work only with the matrix $h(x')$. After a Fourier transform, the Bloch Hamiltonian is given by
\begin{align}
H(k)=\sum_{x\in\Z^d} h(x)e^{i\langle k,x\rangle},
\end{align}
for $k\in T^d$.

We now view $\Z^d$ as being embedded into some bigger lattice $\Z^D$ with $D>d$. In \eqref{hamiltonian}, a canonical embedding is given by letting $x,x'\in\Z^D$ and setting $h_{\alpha\beta}(x')=0$ whenever $x'_i\ne 0$ for $i=d+1,\dots,D$. Physically, this means no hopping into the new $D-d$ directions or, equivalently, stacking of the $d$-dimensional system into these directions. 

To generalize the stacking direction, we introduce an invertible, integer $D$-by-$D$ matrix $A\in\text{GL}_D(\Z)$ and define the stacked Hamiltonian to be given by the replacement $h_{\alpha\beta}(x')\mapsto h_{\alpha\beta}(A^{-1}x')$, corresponding to changing the hopping from the $x'$-direction to the $Ax'$-direction.

Defining the projection $P:T^D\to T^d$ by $P(k_1,\dots,k_D):=(k_1,\dots,k_d)$, the Bloch Hamiltonian of the stacked system can be expressed by the lower-dimensional Bloch Hamiltonian:
\begin{align}
H_\text{stack}(k)&=\sum_{x\in\Z^D} h(A^{-1}x)e^{i\langle k,x\rangle}\nonumber\\
&=\sum_{x\in\Z^D} h(x)e^{i\langle k,Ax\rangle}\nonumber\\
&=\sum_{x\in\Z^D} h(x)e^{i\langle A^T k,x\rangle}\nonumber\\
&=\sum_{x\in\Z^d} h(x)e^{i\langle PA^T k,x\rangle}\nonumber\\
&=H(PA^T k).
\end{align}
The change in $k$-dependence descends to the level of ground state maps. Therefore, given a ground state $\psi:T^d\to C_s$, stacking it in $D$ dimensions according to the matrix $A$ yields a map
\begin{align}
\psi_\text{stack}(k)=\psi(PA^T k).
\end{align}

\subsection{Strong (def. (i)) but stackable}

We are now in a position to analyse another example which will illustrate three points:
\begin{itemize}
\item In general, only non-trivial elements in $[S^d,Y]_{\Z_2}\subset[T^d,Y]_{\Z_2}$ cannot be realized by stacking lower-dimensional systems
\item Strong invariants can "break down" in the non-stable regime
\item Base points make a difference in the non-stable regime
\end{itemize}
The model we consider is a two-dimensional system with three bands, one of which is occupied. We assume that the combination $\mathcal{T}\circ\mathcal{I}=\mathcal{I}\circ\mathcal{T}$ of time-reversal $\mathcal{T}$ (with $\mathcal{T}^2=1$) and inversion $\mathcal{I}$ (with $\mathcal{I}^2=1$) is a symmetry, but individually, both $T$ and $I$ symmetries are broken. Since this symmetry fixes all momenta, the image of any ground state map lies within the fixed point set $\text{Gr}_1(\mathbb{R}^3)\subset\text{Gr}_1(\mathbb{C}^3)$. Therefore, we can consider non-equivariant ground state maps
\begin{align}\label{gsnematics}
\psi:T^2\to\text{Gr}_1(\mathbb{R}^3),
\end{align}
which have been studied in the context of nematics in~\cite{jaenich,bechtluft,chen}, where the torus plays the role of a measuring surface around closed defect lines.

The topological phases in this setting are given by
\begin{align}
[T^2,\text{Gr}_1(\mathbb{R}^3)]&=\{(n_0;n_1,n_2)\mid n_1,n_2\in\Z_2;\nonumber\\
&\qquad n_0\in\mathbb{N}\text{ for }n_0=n_1=0\text{ and }\nonumber\\
&\qquad n_0\in\Z_2\text{ otherwise}\}.\label{nematics}
\end{align}
The strong invariant $n_0\in\mathbb{N}=[S^2,\text{Gr}_1(\mathbb{R}^3)]$ comes from the mapping degree of maps $S^2\to S^2$ composed with the projection $S^2\to S^2/\Z_2=\text{Gr}_1(\mathbb{R}^3)$. In the language of physics, $n_0$ represents the skyrmion charge (see Fig.~\ref{skyrmion}). Note that preserving base points gives $[S^2,\text{Gr}_1(\mathbb{R}^3)]^*=\Z$, on which the fundamental group $[S^1,\text{Gr}_1(\mathbb{R}^3)]^*$ acts via multiplication by $-1$. Hence, after identification to obtain the unbased classes (see appendix B), $\Z$ changes to $\mathbb{N}$ (skyrmions of charge $n_0$ are homotopic to ones with charge $-n_0$).

Again, the strong invariants $\mathbb{N}$ "break down" to $\Z_2$ in the presence of non-trivial lower-dimensional invariants $n_1,n_2\in\Z_2=[S^1,\text{Gr}_1(\mathbb{R}^3)]$, which are non-trivial if they involve a $180°$ rotation of lines along the loop, a configuration known as a Moebius strip (see Fig.~\ref{01} for an example which is canonically embedded into two dimensions).

In this model, \textit{all} classes except those of the form $(n_0;0,0)$ with $n_0\ne0$, which correspond to non-trivial elements in $[S^2,\text{Gr}_1(\mathbb{R}^3)]\subset[T^2,\text{Gr}_1(\mathbb{R}^3)]$, have representatives that are stacked versions of one-dimensional systems~\cite{jaenich,bechtluft,chen}. This gives a total of \textit{seven} stackable classes (four with $n_0=0$ and three with $n_0=1$). This is remarkable, since naively, the $\Z_2$ classification in one dimension would suggest only \textit{four} classes ($\Z_2$ in both linearly independent directions and $n_0=0$). Hence, in general, only definition (ii) can guarantee that no stackable strong topological insulators can occur. Notice that the full restriction of strong topological insulators to $[S^2,\text{Gr}_1(\mathbb{R}^3)]\subset[T^2,\text{Gr}_1(\mathbb{R}^3)]$ is necessary in order to avoid strong topological insulators that may be stacked, excluding the possibility of a compromise between definitions (i) and (ii).

\begin{figure}
\centering
\subfloat[]{\label{skyrmion}\includegraphics[width=0.2\textwidth]{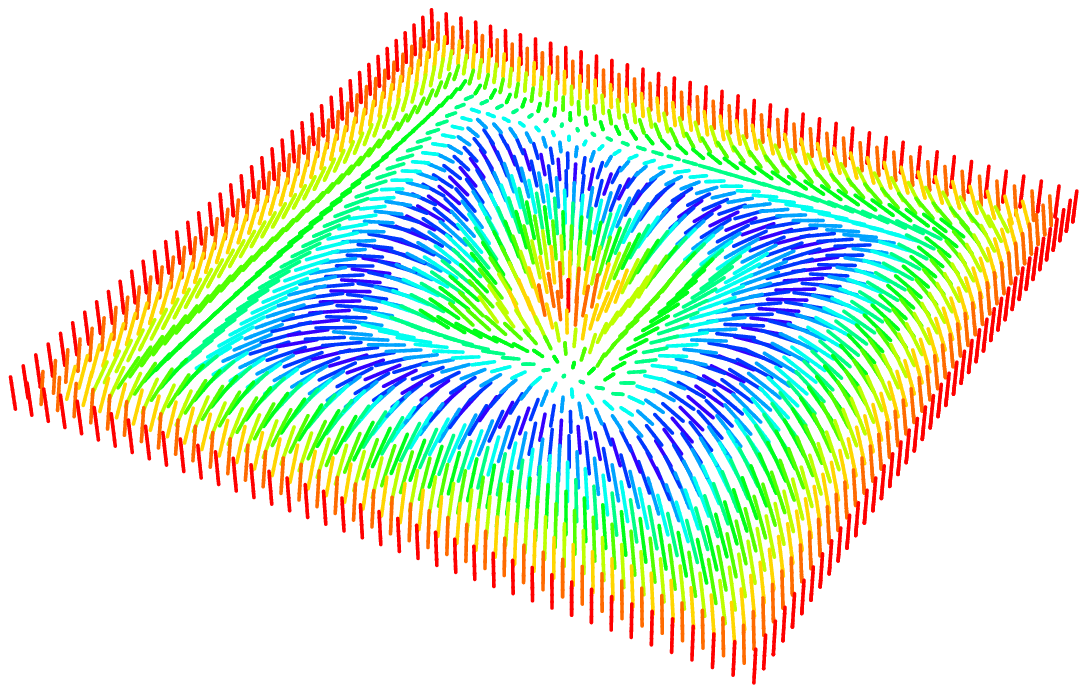}}\quad
\subfloat[]{\label{01}\includegraphics[width=0.2\textwidth]{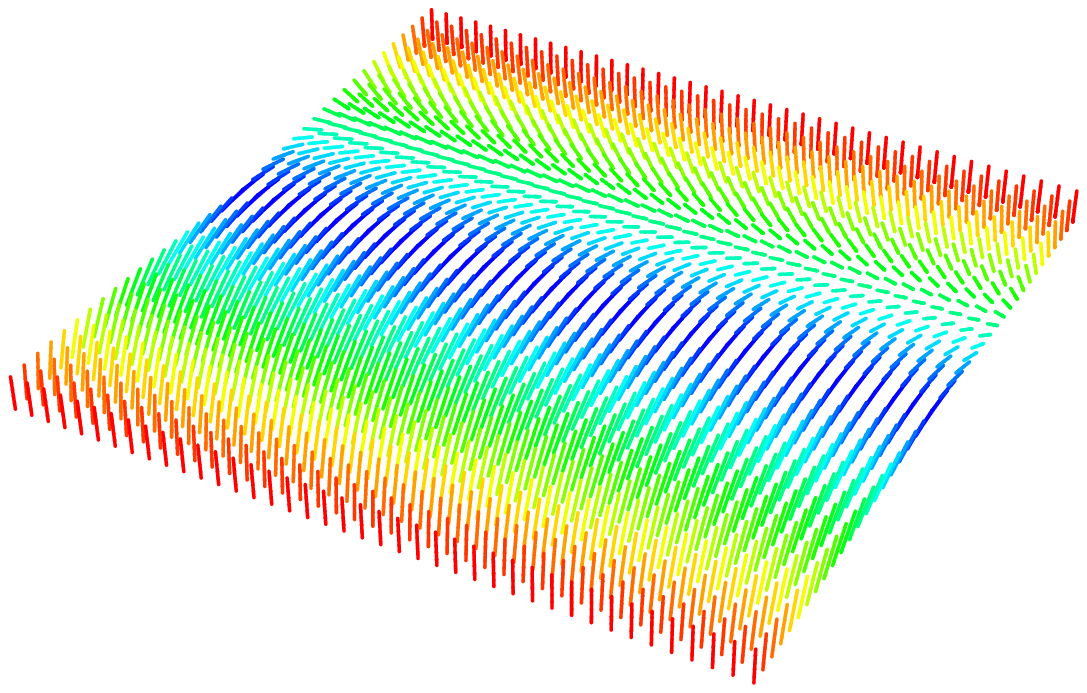}}\\
\subfloat[]{\label{skyrmion01}\includegraphics[width=0.2\textwidth]{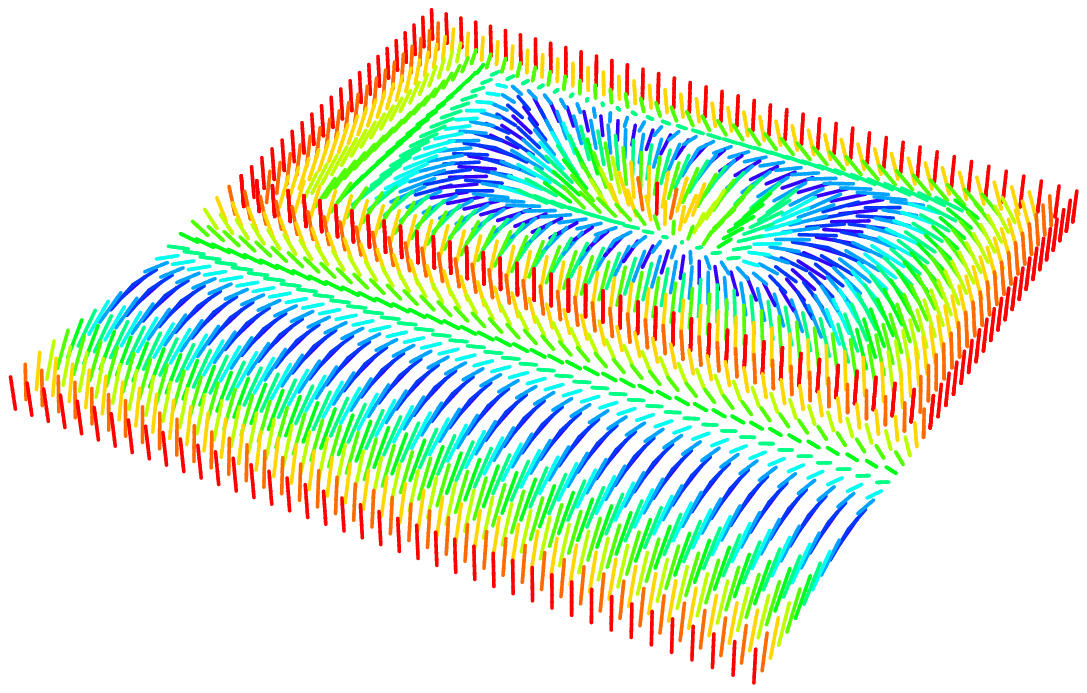}}\quad
\subfloat[]{\label{21}\includegraphics[width=0.2\textwidth]{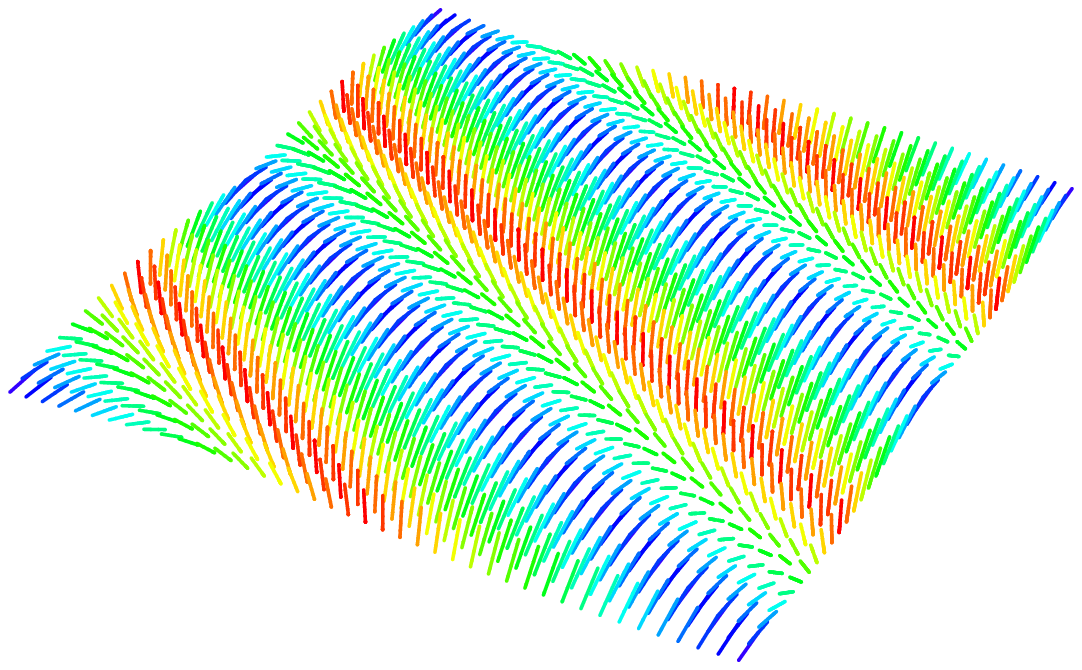}}\\
\subfloat[]{\label{skyrmion11}\includegraphics[width=0.2\textwidth]{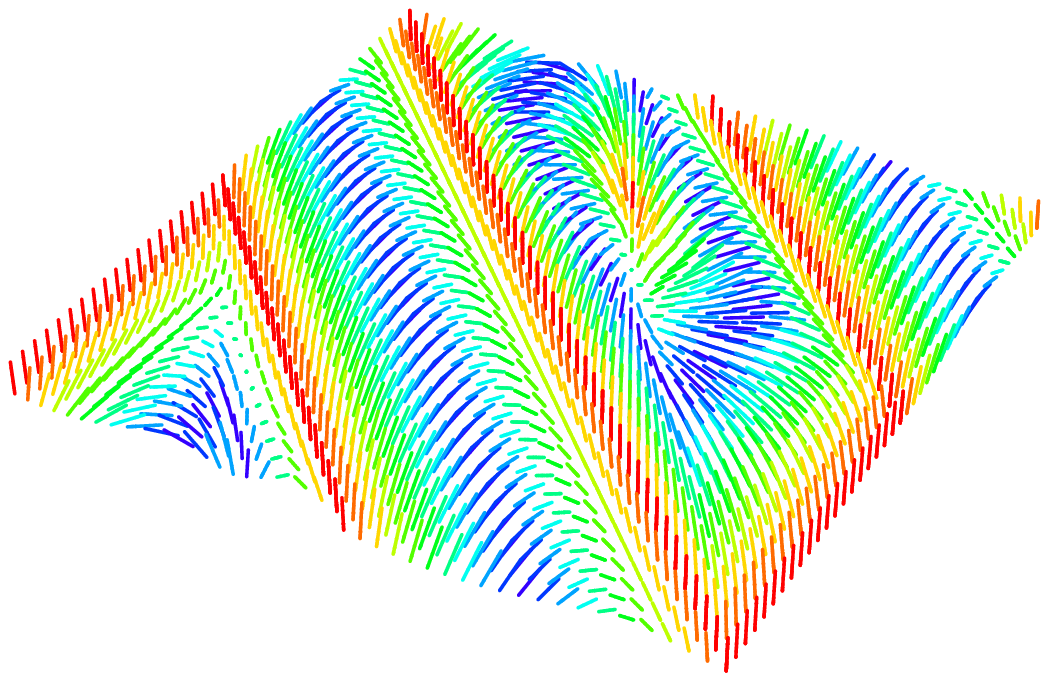}}\quad
\subfloat[]{\label{31}\includegraphics[width=0.2\textwidth]{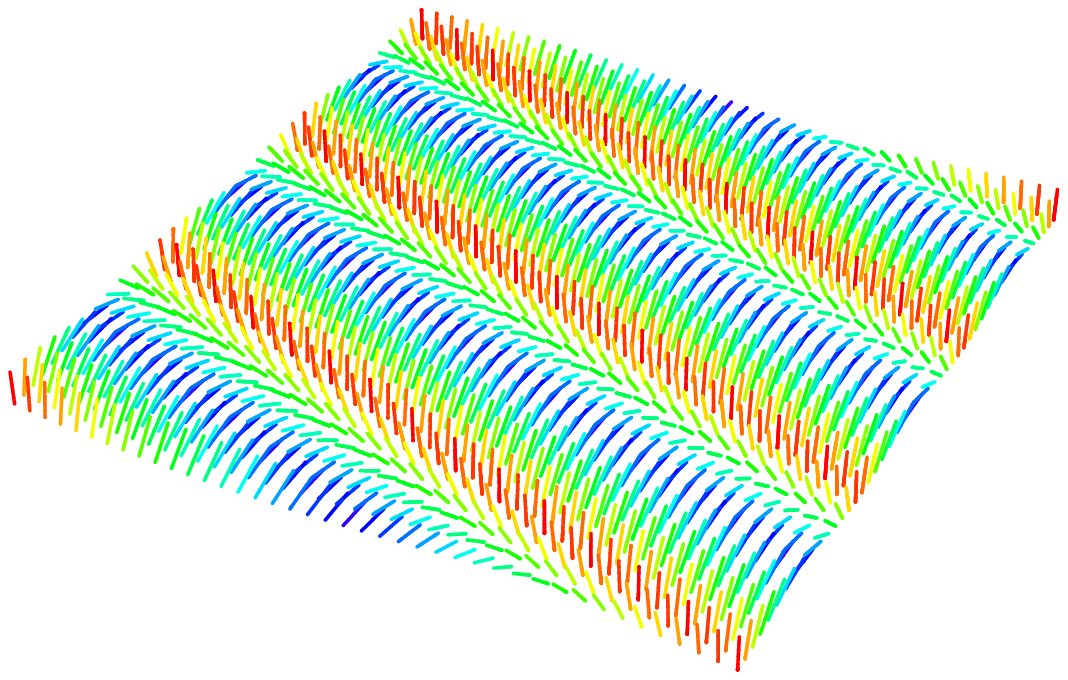}}
\caption[]{Maps $T^2\to\text{Gr}_1(\mathbb{R}^3)$ visualized by placing the image of a point (a line in $\mathbb{R}^3$) on the point itself. $T^2$ is modelled here as a square with periodic boundary conditions. Colours represent the angle to the axis out of the plane. Using the notation $(n_0;n_1,n_2)$, \subref{skyrmion} corresponds to $(1;0,0)$, \subref{01} to $(0;1,0)$, \subref{skyrmion01} and \subref{21} to $(1;1,0)$ and  \subref{skyrmion11} and \subref{31} to $(1;1,1)$. Remarkably, all except \subref{skyrmion} are homotopic to stacked one-dimensional insulators.}\label{skyrmions}
\end{figure}

Fig.~\ref{skyrmions} illustrates some representatives. While a single skyrmion representing $(1;0,0)$ cannot be realized by stacking, having a skyrmion combined with non-trivial projections allows for this possibility.

\subsection{Weak but not stackable}

The following is an example of a weak topological insulator (acording to both definitions) in the stable regime, which cannot be realized through stacking: In two dimensions, consider a $4n$-band model with $2n$ occupied and $2n$ empty bands in class \AIII. Let there be a $U_1$-symmetry, for example conservation of the spin $S_z$-component, which commutes with the chiral operator (which in turn anti-commutes with the Hamiltonian). Then the ground state is a map
\begin{align}
\psi:T^2\to U_n\times U_n
\end{align}
and the topological phases  (homotopy classes) are given by
\begin{align}
[T^2,U_n\times U_n]
&=[T^2,U_n]\times[T^2,U_n]\nonumber\\
&=(\Z\times\Z)\times(\Z\times\Z)\label{4invariants}
\end{align}
Since $[S^2,U_n]=0$ (only weak topological insulators here, according to either definition) and $U_n$ is connected, the factors $\Z$ originate from $[S^1,U_n]=\Z$. Writing $\psi(k_1,k_2)=(\psi_1(k_1,k_2),\psi_2(k_1,k_2))\in U_n\times U_n$, the invariants \eqref{4invariants} are given by the winding numbers of $\det(\psi_i(k_1,0))$ and $\det(\psi_i(0,k_2))$ for $i=1,2$.

One-dimensional versions of this model are classified by $[S^1,U_n\times U_n]=\Z\times\Z$, with invariants given by the winding numbers of $\det(\psi_i(k))$ with $i=1,2$ and $k\in S^1$. Stacking a representative of the class $(n,m)$ according to some matrix $A\in\text{GL}_2(\Z)$ yields an element in the class
\begin{align}
(A_{11}n,A_{21}n)&\times(A_{11}n,A_{21}m)\nonumber\\
\in(\Z\times\Z)&\times(\Z\times\Z)
\end{align}
Clearly, not all classes can be of this form, the simplest counter-example being $(1,0)\times(0,1)$. The mathematical reason is the fact that $\Z\times\Z$ is not generated by a single element. The physical reason is that the non-trivial winding for spin up happens along a linearly independent direction from that of the non-trivial winding for spin down and therefore there is no corresponding one-dimensional system.

\section{Conclusion}

We have proposed to use the natural concept of homotopy theory to extend the results from $K$-theory (and the intermediate step of considering isomorphism classes of vector bundles) to the non-stable regime, which includes \textit{all} insulators independent of the number of occupied/empty bands. We showed that the definition of strong topological insulators has to be more restrictive in general in order to avoid the possibility of realising them by stacking lower-dimensional systems. Furthermore, we demonstrated that there are topological insulators in $d$ dimensions that are "truly $d$-dimensional" despite being weak (independently of how the latter is defined), meaning they cannot be realized by stacking lower-dimensional systems.

Along the way, we derived some useful technical results: We showed how the factorization of topological invariants in the stable regime can be understood from the perspective of homotopy theory and proved that in general (in the stable and non-stable regime), it is legitimate to replace a domain consisting of products of spheres by a single sphere of the same total dimension, even in the presence of $\Z_2$-equivariance.

\section{Acknowledgements}

This work is supported by Deutsche Telekom Stiftung (RK) and Bonn-Cologne Graduate School (RK and CG). RK would like to thank Dominik Ostermayr and Martin Zirnbauer for many useful discussions.

\section{Appendix A: Proof of (\ref{decomposition})}

The crucial feature of the stable regime is a result called \textit{Bott periodicity}~\cite{bott}: Denoting by $\Omega Y$ the space of all based loops in $Y$, it states that there are maps
\begin{align}
B_s^\mathbb{C}:C_s&\to\Omega{C_{s-1}}\\
B_s^\mathbb{R}:R_s&\to\Omega{R_{s-1}},
\end{align}
inducing isomorphisms
\begin{align}
B_{s*}^\mathbb{C}:[S^d,C_s]^*&\stackrel{\sim}{\to}[S^d,\Omega{C_{s-1}}]^*=[S^{d+1},C_{s-1}]^*\label{eq:bottC}\\
B_{s*}^\mathbb{R}:[S^d,R_s]^*&\stackrel{\sim}{\to}[S^d,\Omega{R_{s-1}}]^*=[S^{d+1},R_{s-1}]^*,\label{eq:bottR}
\end{align}
for all $1\le d\le d_\text{max}$, where $d_\text{max}$ depends on the symmetry class and increases monotonously with $m$ and $n$. In the stable limit $m,n\to\infty$, we have $d_\text{max}\to\infty$ and $B_s^\mathbb{C}$ as well as $B_s^\mathbb{R}$ become what is called weak homotopy equivalences. The case $d=0$ may be included for $B_{s*}^\mathbb{C}$ if $s$ is odd and for $B_{s*}^\mathbb{R}$ if $s\ne2,6$. The reason is that the spaces $C_s$ ($s$ even), $R_2$ and $R_6$ have $n-1$ connected components (see Table~\ref{periodic}), whereas the corresponding right hand sides of eqs.~\eqref{eq:bottC} and~\eqref{eq:bottR} are isomorphic to $\Z$. The set $[S^d,Y]^*$ can be equipped with a group structure for $d\ge1$, which is given by concatenation of loops for $d=1$ and a similar construction using only one coordinate for $d>1$~\cite{hatcher}. These groups are known as the homotopy groups $\pi_d(Y)$ and the group structure will be important in the following.

Bott periodicity sets apart the stable from the non-stable regime and therefore, not surprisingly, will be central to our proof. This distinction is to be expected from the non-stable examples in the main text, for which \eqref{decomposition} does not hold.

The torus $T^d$ is a $\Z_2$-CW complex~\cite{tomdieck2,shah,may} (a space with $\Z_2$-action built by inductively attaching disks of increasing dimensions along their boundary spheres) and as such the relations between homotopy groups give a lot of information about the relations between the sets of topological phases in the different classes. This statement is formalized by the equivariant Whitehead Theorem~\cite{tomdieck2,shah,may}, which states that the fact that $B_s^\C$ is equivariant and restricts to $B_s^\R$ (see~\cite{quast}), both of which induce bijections on homotopy groups for $d\le d_\text{max}$ and odd $s$, implies that $B_s^\C$ also induces a bijection
\begin{align}\label{whitehead}
B_{s*}^\mathbb{C}:[T^d,C_s]_{\Z_2}\stackrel{\sim}{\to}[T^d,\Omega C_{s-1}]_{\Z_2},
\end{align}
for $d<d_\text{max}$ and odd $s$.

The last ingredient needed in order to further evaluate the right hand side of \eqref{whitehead} is the equivariant free loop fibration, which requires the introduction of some notation: If $g\in\Z_2$ is the non-trivial element, then for a space $Y$ with $\Z_2$-action, we denote by $Y^{\Z_2}$ the subset that is fixed under $g$. An important space is the space of free (i.e. unbased) loops in a space $Y$, denoted by $LY$, which is equipped with the $\Z_2$-action $f\mapsto g\circ f\circ g$ for $f:S^1\to Y$ (with the action $g\cdot\phi=-\phi$ on the angle coordinate of $S^1$). The space of equivariant free loops is then given by $(LY)^{\Z_2}$. The equivariant free loop fibration is a map $p:(LY)^{\Z_2}\to Y^{\Z_2}$ assigning to a free equivariant loop $S^1\to Y$ its value at a fixed point $s_0\in (S^1)^{\Z_2}=S^0$. The fiber at a point $y_0\in Y^{\Z_2}$ is $(\Omega Y)^{\Z_2}$ (based equivariant loops at $y_0$). Importantly, this fibration has a section $q:Y^{\Z_2}\to (LY)^{\Z_2}$ given by assigning to $y_0\in Y^{
\Z_2}$ the constant loop at $y_0$, which yields $p\circ q=\mathrm{Id}_{Y^{\Z_2}}$.

All fibrations provide a long exact sequence of homotopy groups. In the case of the equivariant free loop fibration, the existence of the section $q$ implies that this sequence splits into short exact sequences
\begin{align}
0\to\pi_k((\Omega Y)^{\Z_2})\stackrel{i_*}{\to}\pi_k((LY)^{\Z_2})\overset{p_*}{\underset{q_*}{\rightleftarrows}}\pi_k(Y^{\Z_2})\to0\label{short}
\end{align}
for all $k\ge0$, where $i_*$ is induced by the inclusion $i$ of the fiber $(\Omega Y)^{\Z_2}$ into $(LY)^{\Z_2}$. As stated previously, $\pi_k$ has a group structure for $k\ge1$ and in that case, all maps in~\eqref{short} are homomorphisms. In this situation, every element of $\pi_k((LY)^{\Z_2})$ can be written uniquely as a sum $i_*[a]+s_*[b]$ for $[a]\in\pi_k((\Omega Y)^{\Z_2})$ and $[b]\in\pi_k(Y^{\Z_2})$. In other words,
\begin{align}
\pi_k((LY)^{\Z_2})&\stackrel{\text{as sets}}{\simeq}\pi_k((\Omega Y)^{\Z_2})\times\pi_k(Y^{\Z_2})
\end{align}
The fact that this decomposition only works for $k\ge1$ is crucial: The left hand side of \eqref{whitehead}, which we want to evaluate in the end, is the same as $\pi_0((L^d C_s)^{\Z_2})$. Here $(L^d C_s)^{\Z_2}$ is the $d$-fold iterated equivariant free loop space of $C_s$, which is the space of equivariant maps from $T^d$ to $C_s$ and $\pi_0$ is the set of its connected components. In other words, it is the set of $d$-dimensional topological phases in symmetry class $s$. Eqs. \eqref{hopf} and \eqref{nematics} show that the decomposition cannot work in general at the level of $\pi_0$ and consequently we  really need to go to the loop space $\Omega C_{s-1}$ in order to arrive at the case $k=1$ and complete the proof:
\begin{align}
[T^d,C_s]_{\Z_2}&\simeq[T^d,\Omega C_{s-1}]_{\Z_2}\nonumber\\
&\simeq\pi_1((L^d C_{s-1})^{\Z_2})\nonumber\\
&\simeq\pi_1((L^{d-1}\Omega C_{s-1})^{\Z_2})\times\pi_1((L^{d-1} C_{s-1})^{\Z_2})\nonumber\\
&\,\,\,\vdots\nonumber\\
&\simeq\prod_{l=0}^d(\pi_{1}((\Omega^{l}C_{s-1})^{\Z_2}))^{d \choose l}\label{eq:factorized}\\
&\simeq\prod_{l=0}^d(\pi_{0}((\Omega^{l}C_{s})^{\Z_2}))^{d \choose l}\nonumber\\
&\simeq\prod_{l=0}^d\Big([S^l,C_s]^*_{\Z_2}\Big)^{d \choose l}\nonumber
\end{align}
In the second to last equality, the equivariant Whitehead theorem was used again (this time in its base-point preserving version) in order to arrive at a result for the target space $C_s$. This completes the proof for both the real and the complex classes (the latter are included by choosing trivial $\Z_2$-actions) with odd $s$.

For even $s$, the requirements for the equivariant Whitehead theorem are not met since the complex Bott maps $B_{s*}^\C$ in eq.~\eqref{eq:bottC} are not a bijection for $d=0$. This shortcoming is remedied by replacing $C_s$ by its connected component $(C_s)_0$ containing the base point as well as $\Omega C_{s-1}$ by its connected component $(\Omega C_{s-1})_0$ containing the constant loop at the base point of $C_{s-1}$. The equivariant Whitehead theorem then gives a bijection
\begin{align}
[T^d,(C_s)_0]_{\Z_2}\simeq[T^d,(\Omega C_{s-1})_0]_{\Z_2}.
\end{align}
The right hand side of this equation is a subset of $[T^d,\Omega C_{s-1}]_{\Z_2}$. It can be identified in the decomposition in~\eqref{eq:factorized} as the subset with the factor $\pi_1(C_{s-1}^{\Z_2})=\pi_1(R_{s-1})$ replaced by $\ker((i_{s-1})_*)\subset\pi_1(R_{s-1})$, where
\begin{align}
(i_{s-1})_*:\pi_1(R_{s-1})\to\pi_1(C_{s-1})
\end{align}
is the induced map of the inclusion $i_{s-1}:R_{s-1}\hookrightarrow C_{s-1}$.

For the real classes with $s\neq2,6$, $\ker((i_{s-1})_*)=\pi_1(R_{s-1})$ and $[T^d,C_s]_{\Z_2}=[T^d,(C_s)_0]_{\Z_2}$, where latter follows from the observation that for $s\neq2,6$, $R_s\subset(C_s)_0$ and therefore the image of equivariant maps from $T^d$ is always contained in $(C_s)_0$. Thus, the result in these cases is equivalent to the result for odd $s$.

In the symmetry classes A, \AII ($s=2$) and \AI ($s=6$), the set $\ker((i_{s-1})_*)$ contains only one element and $[T^d,C_s]_{\Z_2}\neq[T^d,(C_s)_0]_{\Z_2}$, so the result needs to be modified as stated below eq.~\eqref{decomposition}.

\section{Appendix B: Proof of (\ref{inclusion1}) and (\ref{inclusion2})}

In this section we give a proof of the following theorem (eqs.~\eqref{inclusion1} and \eqref{inclusion2}):
\begin{thm}\label{theorem1}
\begin{align*}
[S^d,Y]_{\Z_2}&\subset[T^d,Y]_{\Z_2}\\
[S^{d_x+d_k},Y]_{\Z_2}&\subset[S^{d_x}\times T^{d_k},Y]_{\Z_2}
\end{align*}
\end{thm}
The first line allows defining \textit{strong} topological insulators as non-trivial elements of the left hand side, while the second line shows that in the presence of defects with codimension $d_x+1$, one may replace the product of sphere and torus by a single sphere of equal total dimension, at the cost of potentially missing non-trivial classes.

As a model for maps from both torus and sphere, we will use the $d$-dimensional cube $[-\pi,\pi]^d$ as the domain, with coordinates $-\pi\le x_i\le\pi$, $i=1,\dots,d$. Maps from the torus are realized by requiring them to be periodic in all directions (same values on opposing sides of the cube), while maps from the sphere are required to map the entire boundary of the cube to a single point. The action of $\Z_2$ on the cube is given coordinate-wise as either $x_i\mapsto-x_i$ (non-trivial or momentum-like) or $x_i\mapsto+x_i$ (trivial or position-like), which gives a total of $2^d$ possible actions.

A crucial ingredient in the proof is the $\Z_2$-equivariant version of the relation between homotopy classes of maps with fixed basepoints, denoted $[S^d,Y]_{\Z_2}^*$, and those without fixed basepoints, denoted $[S^d,Y]_{\Z_2}$. If the fixed point set $Y^{\Z_2}$ is connected, then any representative of an unbased class can be homotoped to a based map (we assume that the basepoint $y_0\in Y^{\Z_2}$). However, two based maps that represent different elements in $[S^d,Y]_{\Z_2}^*$ may represent the same element in $[S^d,Y]_{\Z_2}$, meaning there can be an unbased homotopy even though no based one exists. This unbased homotopy takes the image of the basepoint $s_0\in(S^d)^{\Z_2}$ to a loop in $Y^{\Z_2}$ and identifying all based maps up to these loops gives a bijection~\cite{tomdieck,whiteheadbook}
\begin{align}\label{based}
[S^d,Y]_{\Z_2}\simeq [S^d,Y]_{\Z_2}^*/[S^1,Y^{\Z_2}]^*.
\end{align}
If $Y^{Z_2}$ is disconnected, we denote by $Y^{Z_2}_0$ the component containing the base point and introduce the notation $[(X,x_0),(Y,Z)]_{\Z_2}$ for equivariant homotopy classes of maps $X\to Y$ which map $x_0\in X$ to $Z\subset Y$. For example, given base points $s_0\in S^d$ and $y_0\in Y^{\Z_2}\subset Y$, we have $[S^d,Y]^*_{\Z_2}=[(S^d,s_0),(Y,\{y_0\})]_{\Z_2}$. Then eq.~\eqref{based} holds in the following, modified form:
\begin{align}\label{based2}
[(S^d,s_0),(Y,Y^{\Z_2}_0)]_{\Z_2}\simeq [S^d,Y]_{\Z_2}^*/[S^1,Y^{\Z_2}]^*.
\end{align}
In fact, eqs.~\eqref{based} and \eqref{based2} hold more generally with $S^d$ replaced by any $\Z_2$-CW complex, though this will not be needed in the present paper.

The identifications~\eqref{based} and \eqref{based2} have a simple geometrical interpretation: The boundary of $S^d=[-\pi,\pi]^d$ is always a fixed point of the $\Z_2$-action and as such it has to map to $Y^{\Z_2}$. A loop $\gamma$ representing a class in $[S^1,Y^{\Z_2}]^*$ now acts on a representative $f$ of a class in $[S^d,Y]_{\Z_2}^*$ by moving the image point of the boundary along $\gamma$ to give a map $b_d(\gamma,f):S^d\to Y$ (see Fig. \ref{action}). In formulas,
\begin{align}\label{bd}
b_d(\gamma,f)(x):=
\begin{cases}
f(2x)&\text{ for }|x|\le \frac{\pi}{2}\\
\gamma(3\pi-4|x|)&\text{ for }|x|>\frac{\pi}{2},
\end{cases}
\end{align}
where $|x|:=\max(x_i)_{i=1\dots d}$.

\begin{figure}
\centering
\subfloat[]{\label{action1d}\includegraphics[width=0.2\textwidth]{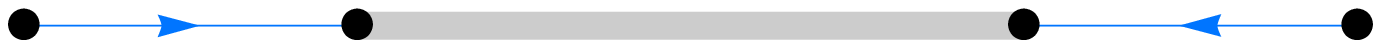}}\quad
\subfloat[]{\label{action2d}\includegraphics[width=0.2\textwidth]{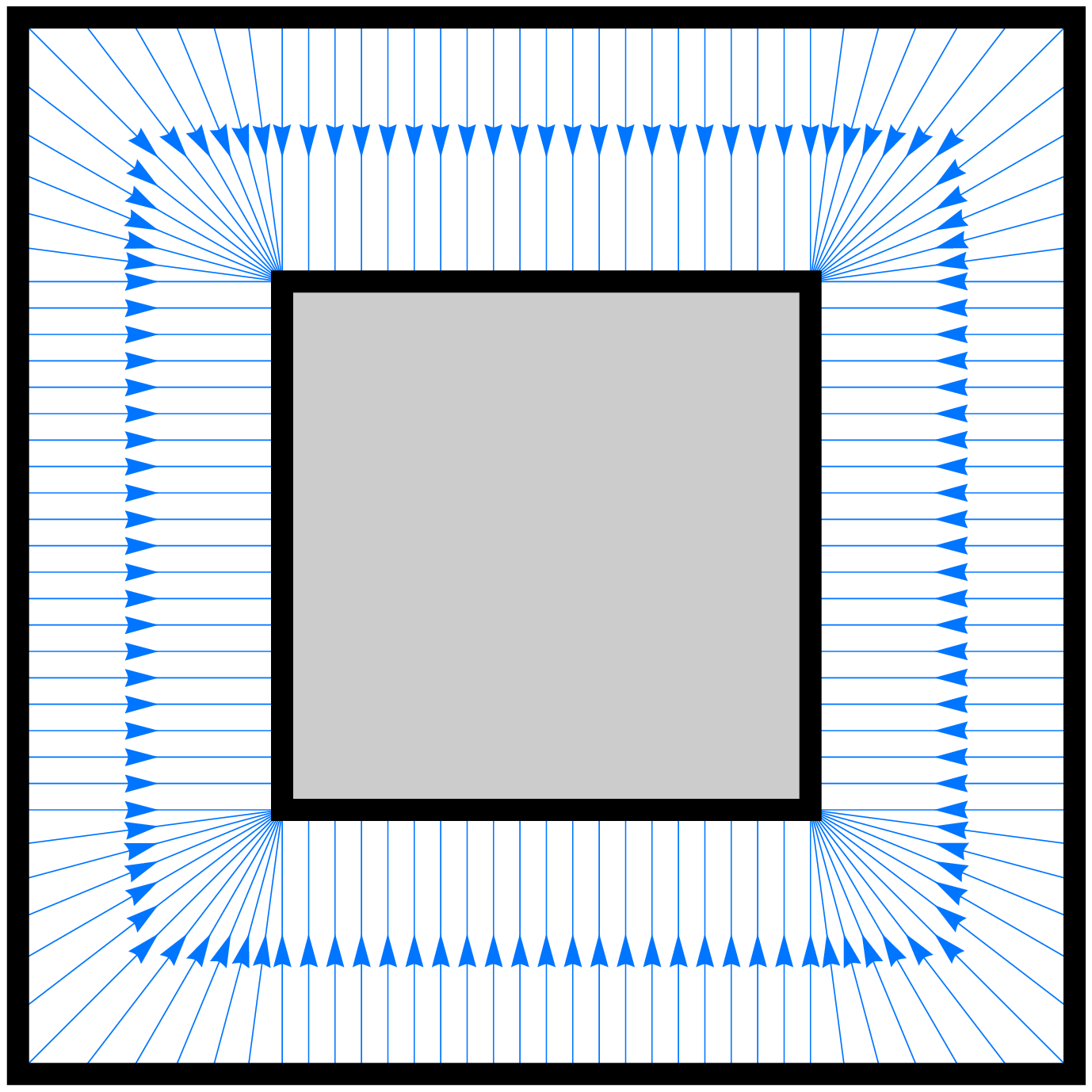}}\\
\subfloat[]{\label{action3d}\includegraphics[width=0.2\textwidth]{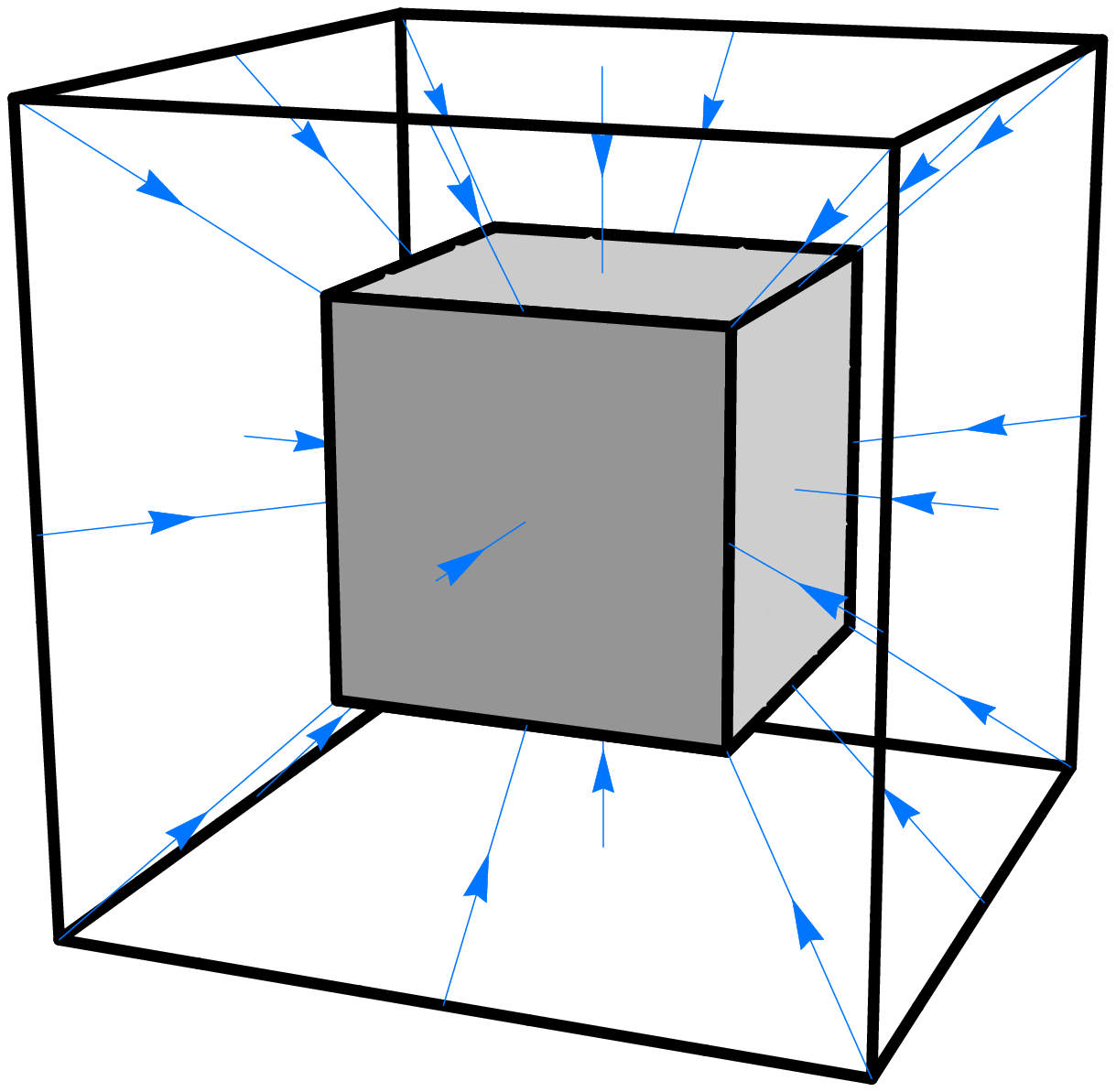}}\quad
\caption[]{The domain of $b_d(\gamma,f)$ for \subref{action1d} $d=1$, \subref{action2d} $d=2$ and \subref{action3d} $d=3$. The loop $\gamma$ is represented in blue with an arrow indicating the direction in which it is traversed and the domain of $f$ is depicted in gray. In \subref{action1d} and \subref{action2d}, black points are mapped to the base point $y_0\in Y^{\Z_2}$. In \subref{action3d}, the entire surfaces of the two cubes are mapped to $y_0$.}\label{action}
\end{figure}

Although defined on the level of representatives, eq.~\eqref{bd} yields a well-defined action on the level of homotopy classes and the orbit of this action is identified on the right hand side of~\eqref{based}. In the following special case, the map $b_d$ simplifies considerably, which will later be crucial for the proof of the theorem:

\begin{lemma}\label{lemma1}
For $[\gamma]\in[S^1,(LY)^{\Z_2}]^*$ and $[f]\in[S^d,\Omega Y]_{\Z_2}^*$,
\begin{align}
[b_d(\gamma,f)]=[b_{d+1}(\gamma(\cdot)(0),f)]\text{ in }[S^d,LY]_{\Z_2}^*.
\end{align}
\end{lemma}
On the right hand side of the equation, $f$ is interpreted as a map $S^{d+1}\to Y$.
\begin{proof}

The map $\gamma$ is a based loop of free loops with base point being the constant loop at $y_0\in Y$. Alternatively, it may be viewed as a free loop of based loops by switching the two loop coordinates. The latter interpretation is shown in Fig. \ref{initial} for $d=1$, where lines with arrows represent based loops. The fact that this is a free loop of based loops is indicated by the colour code: All these loops may be different, but there are periodic boundary conditions (the most upper based loop is the same as the lowest one, both being shown in orange). 

The map $b_d(\gamma,f)(\cdot,\pm\pi)$ is homotopic to $f(\cdot)(\pm\pi)$, since $f(x)(\pm\pi)=y_0$ and the action fixes the neutral element. This can be seen in Fig.~\ref{initial} for $d=1$: The upper and lower boundaries correspond to the concatenation of the based loop $\gamma(\cdot)(\pm\pi)$ (orange), the constant loop $f(\cdot)(\pm\pi)$ (black) and the reversed version of $\gamma(\cdot)(\pm\pi)$ (orange, reversed arrow). This combination is clearly homotopic to the constant loop and this homotopy is used to arrive at Fig.~\ref{alpha}.

\begin{figure}
\centering
\subfloat[]{\label{initial}\includegraphics[width=0.2\textwidth]{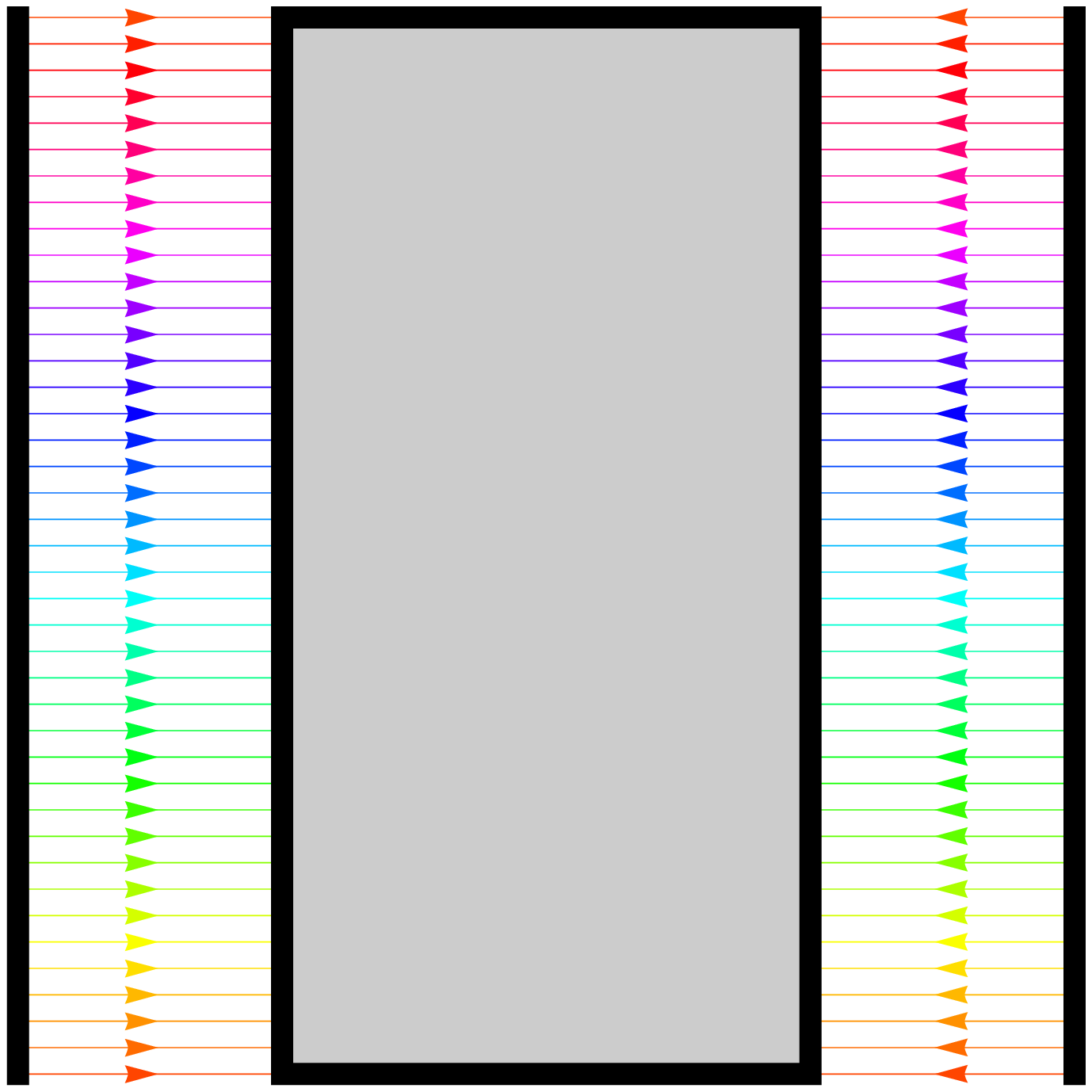}}\quad
\subfloat[]{\label{alpha}\includegraphics[width=0.2\textwidth]{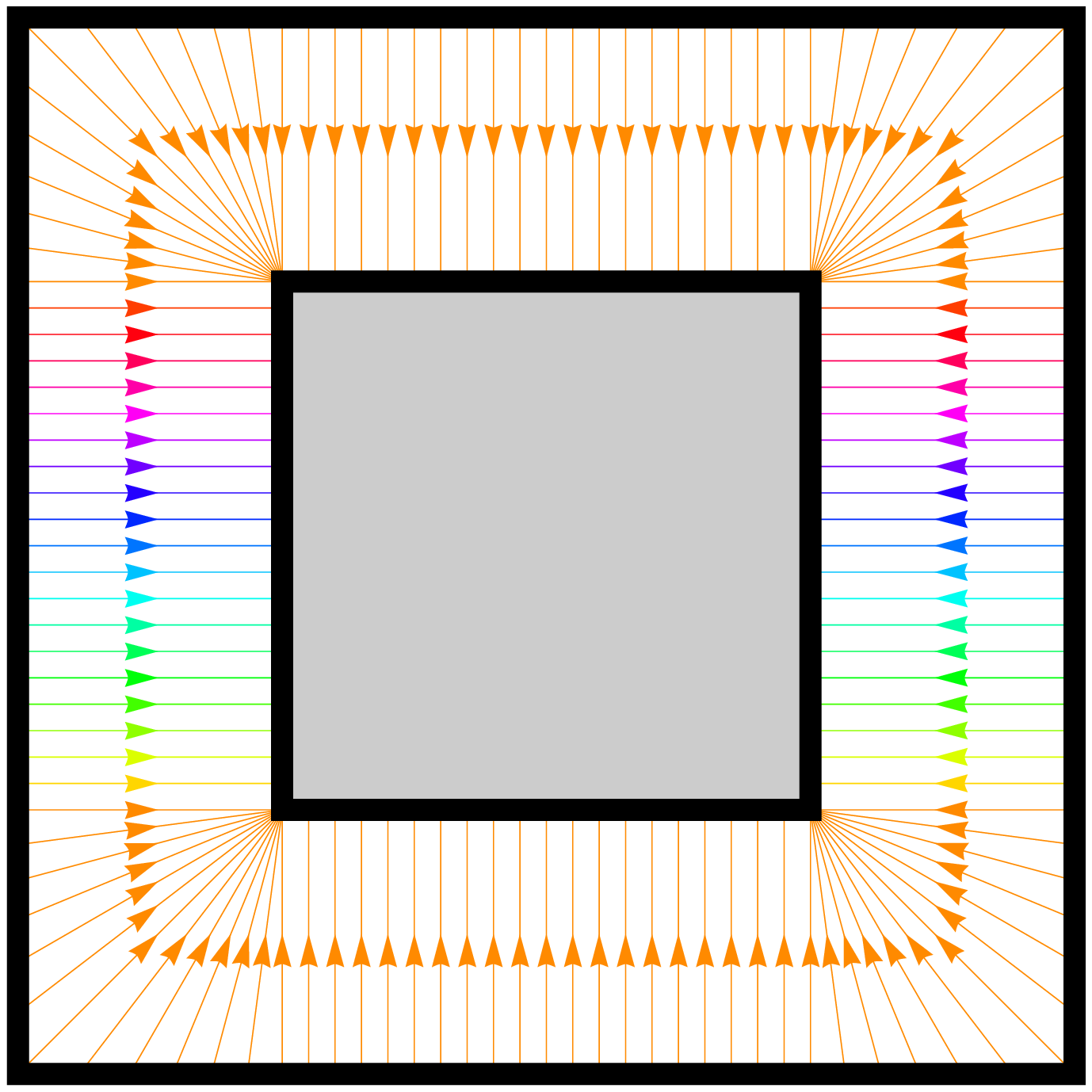}}\\
\subfloat[]{\label{beta}\includegraphics[width=0.2\textwidth]{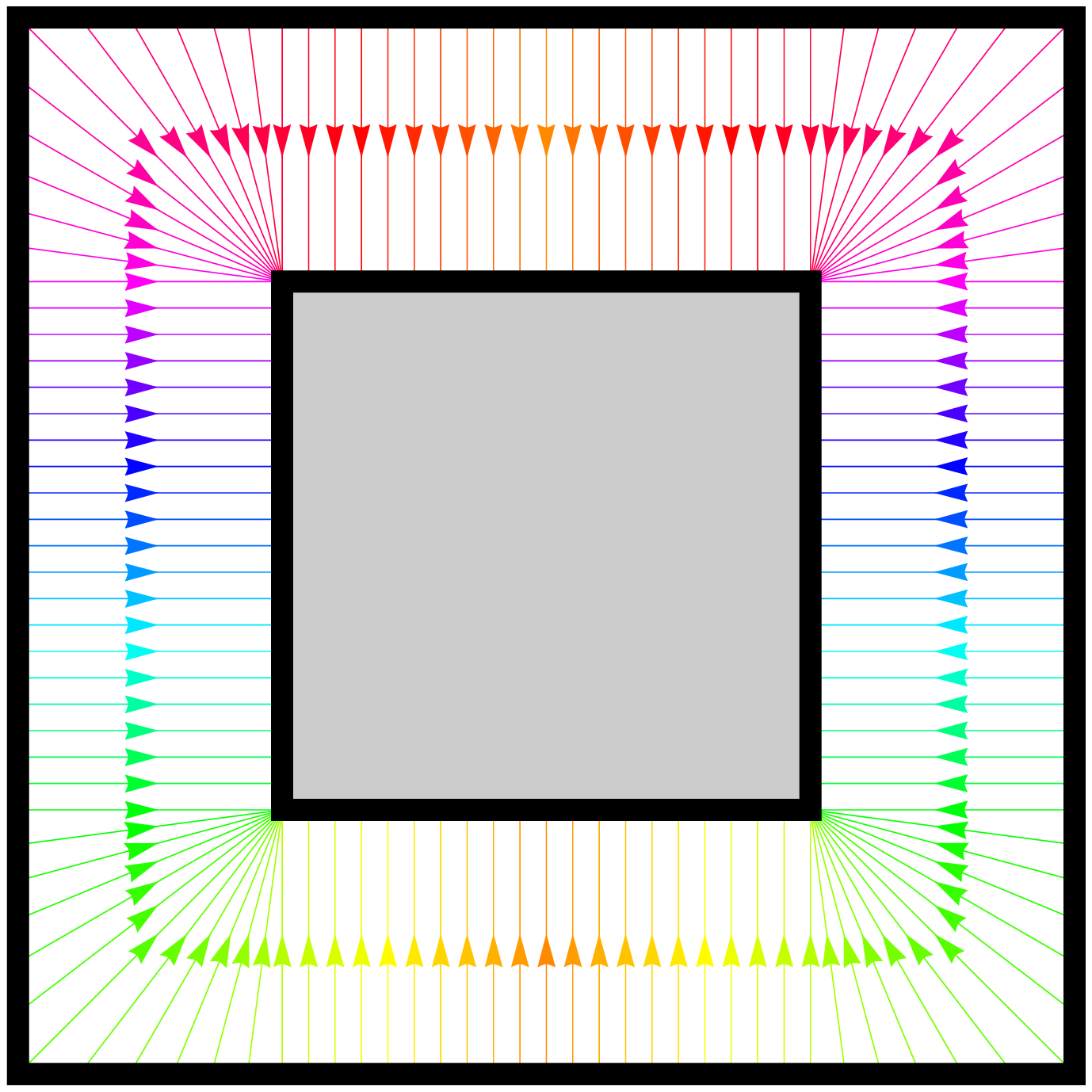}}\quad
\subfloat[]{\label{final}\includegraphics[width=0.2\textwidth]{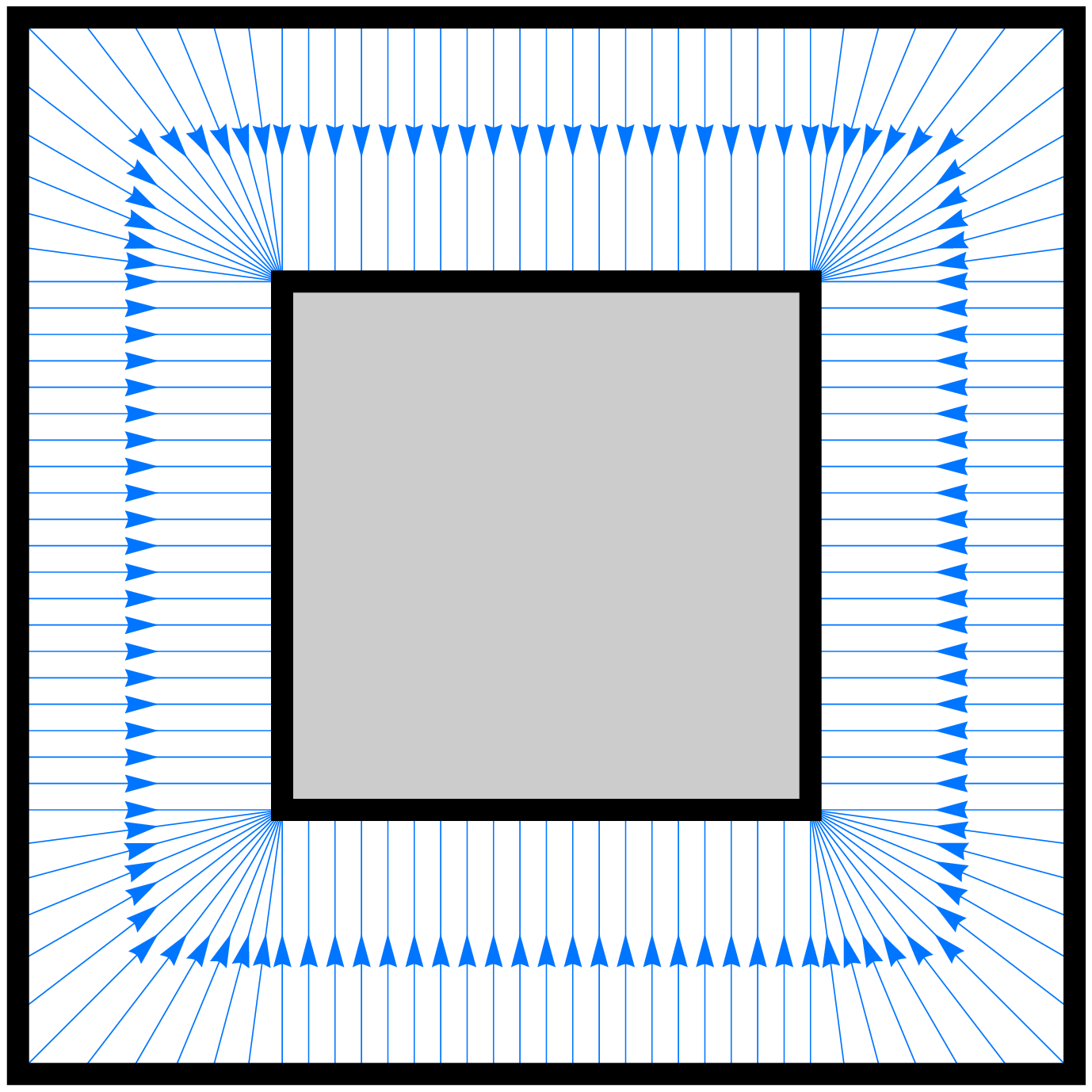}}
\caption[]{Steps in the proof of lemma \ref{lemma1} for $d=1$. The gray area corresponds to the domain of $f:S^1\to \Omega Y$ interpreted as a map $S^2\to Y$. All black lines are mapped to the base point of $Y$. \subref{initial} shows the domain of $b_1(\gamma,f)$, in this case given by conjugation of $f$ by $\gamma:S^1\to L_{\Z_2}Y$. The latter can be viewed as a free loop of based loops (coloured lines) and arrows indicate the direction in which the based loops are traversed. \subref{alpha} shows the result of applying the homotopy of the upper and lower sides to the constant map, giving the configuration with $\alpha_0$. The stage at $\alpha_1=\beta_0$ is shown in \subref{beta}, while \subref{final} depicts the final configuration with $\beta_1$, which corresponds to the domain of $b_{2}(\gamma(\cdot)(0),f)$.}\label{homotopies}
\end{figure}

For the next homotopies, the central part of the cube $[-\pi,\pi]^{d+1}$ which is associated with $f$ (gray area in Fig.~\ref{homotopies}) will remain invariant. The surrounding part is equivalent to a map $S^d\to \Omega Y$, but since we will only use special homotopies that leave the part with $x_{d+1}=0$ invariant (the blue loops in Fig.~\ref{alpha}), we will restrict to only one hemisphere of $S^d$, which is a disk $D^d$. The same homotopies will be applied to the other hemisphere.

We introduce the radial coordinate $0\le r\le1$ of $D^d$, which corresponds to $x_1=\dots=x_d=0$ at $r=0$ and to $x_{d+1}=0$ at $r=1$. The result of using the null homotopy of $b_d(\gamma,f)(\cdot,\pm\pi)$ is a map $\alpha_0:D^d\to \Omega Y$ depicted for $d=1$ in Fig.~\ref{alpha} and given in general by
\begin{align}
\alpha_0(r):=
\begin{cases}
\gamma(\pi)&\text{ for }r\le\frac{1}{2}\\
\gamma(2\pi(1-r))&\text{ for }r>\frac{1}{2},
\end{cases}
\end{align}
The first homotopy is given by
\begin{align}
\alpha_t(r):=
\begin{cases}
\gamma(\pi)&\text{ for }r\le\frac{1-t}{2}\\
\gamma(\frac{2\pi}{1+t}(1-r))&\text{ for }r>\frac{1-t}{2},
\end{cases}
\end{align}
where $0\le t\le1$. For $d=1$, this corresponds to the step from Fig.~\ref{alpha} to Fig.~\ref{beta}: The former shows $\alpha_0:D^1\to\Omega Y$, which maps to the orange loop at $r=0$ and to the blue loop at $r=1$. The homotopy $\alpha_t$ pushes the orange region completely to $r=0$ while "stretching" the remainder accordingly, which results in $\alpha_1$ shown in Fig.~\ref{beta}. 

Subsequently, all other loops are also pushed to $r=0$ and "annihilate", leaving only the blue one. In formulas, this second homotopy is given by
\begin{align}
\beta_t(r):=\gamma(\pi(1-r)(1-t)),
\end{align}
where $\beta_0=\alpha_1$.
Since all $\Z_2$-actions introduced for $[-\pi,\pi]^{d+1}$ fix the radial coordinate $r$ and all homotopies depend only on $r$, they all go through equivariant maps.
\end{proof}
For the next  lemma, we use lemma~\ref{lemma1} to show that the homotopy classes of maps with periodic boundary conditions in one coordinate of $[-\pi,\pi]^d$ include the classes of maps that map to a fixed point at the edges of that interval.
\begin{lemma}
\begin{align}
[(S^d,s_0),(LY,(LY)^{\Z_2}_0)]_{\Z_2}\supset[(S^{d+1},s_0),(Y,Y^{\Z_2}_0)]_{\Z_2}
\end{align}
\end{lemma}
\begin{proof}
\begin{align}
&[(S^d,s_0),(LY,(LY)^{\Z_2}_0)]_{\Z_2}\nonumber\\
&=[S^d,LY]^*_{\Z_2}/[S^1,(LY)^{\Z_2}]^*\label{unbasedbased}\\
&=[S^1,\Omega^dY]_{\Z_2}/[S^1,(LY)^{\Z_2}]^*\label{switch}\\
&\supset[(S^1,s_0),(\Omega^dY,(\Omega^d Y)^{\Z_2}_0)]_{\Z_2}/[S^1,(LY)^{\Z_2}]^*\label{inclusion}\\
&=\Big([S^1,\Omega^dY]^*_{\Z_2}/[S^1,(\Omega^d Y)^{\Z_2}]^*\Big)/[S^1,(LY)^{\Z_2}]^*\label{doublequotient}\\
&=[S^1,\Omega^dY]^*_{\Z_2}/[S^1,Y^{\Z_2}]^*\label{absorbedquotient}\\
&=[(S^{d+1},s_0),(Y,Y^{\Z_2}_0)]_{\Z_2}
\end{align}
This chain of equalities and inclusions needs some explanation: We first use the relation~\eqref{based2} between based and unbased homotopy classes to arrive at~\eqref{unbasedbased}. Then, for eq.~\eqref{switch}, the perspective is changed to viewing the (free) loop parameter of $LY$ as the domain and the $d$ coordinates of $S^d$ as the domain of elements in $\Omega^d Y$. Importantly, this effects a change from based homotopy classes to unbased ones. The inclusion \eqref{inclusion} is well defined on the quotient since $(\Omega^d Y)^{\Z_2}_0$ is fixed under conjugation by elements in $(LY)^{\Z_2}$. Having arrived at~\eqref{doublequotient} by again using~\eqref{based2}, we use lemma~\ref{lemma1} to homotope the action of elements in $[S^1,(\Omega^d Y)^{\Z_2}]^*$ as well as $[S^1,(LY)^{\Z_2}]^*$ to the action of some element in $[S^1,Y^{\Z_2}]^*$, yielding~\eqref{absorbedquotient}. In the last step, we use~\eqref{based2} again to complete the proof.
\end{proof}

We are now equipped to prove theorem \ref{theorem1}: For the case without defects, if $Y^{\Z_2}$ is connected,
\begin{align}
[T^d,Y]_{\Z_2}
&=[S^1,L^{d-1}Y]_{\Z_2}\nonumber\\
&\supset[(S^1,s_0),(L^{d-1}Y,(L^{d-1}Y)^{\Z_2}_0)]_{\Z_2}\nonumber\\
&\supset[(S^2,s_0),(L^{d-2}Y,(L^{d-2}Y)^{\Z_2}_0)]_{\Z_2}\nonumber\\
&\supset\cdots\nonumber\\
&\supset[(S^{d-1},s_0),(LY,(LY)^{\Z_2}_0)]_{\Z_2}\nonumber\\
&\supset[(S^d,s_0),(Y,Y^{\Z_2}_0)]_{\Z_2}\nonumber\\
&=[S^d,Y]_{\Z_2}
\end{align}
If $Y^{\Z_2}$ has several components $Y^{\Z_2}_n$, we repeat the above steps for different base points $y_0\in Y^{\Z_2}_n$ to obtain
\begin{align}
[T^d,Y]_{\Z_2}
&=\coprod_n [(T^d,s_0),(Y,Y^{\Z_2}_n)]_{\Z_2}\nonumber\\
&\supset\coprod_n [(S^d,s_0),(Y,Y^{\Z_2}_n)]_{\Z_2}\nonumber\\
&=[S^d,Y]_{\Z_2}\label{nodefect}
\end{align}
In the presence of defects, similar steps lead to the result of theorem \ref{theorem1}. Assuming again that $Y^{\Z_2}$ is connected,
\begin{align}
&[S^{d_x}\times T^{d_k},Y]_{\Z_2}\nonumber\\
&=[S^{d_x},L^{d_k}Y]_{\Z_2}\nonumber\\
&\supset[(S^{d_x},s_0),(L^{d_k}Y,(L^{d_k}Y)^{\Z_2}_0)]_{\Z_2}\nonumber\\
&\supset[(S^{d_x+1},s_0),(L^{d_k-1}Y,(L^{d_k-1}Y)^{\Z_2}_0)]_{\Z_2}\nonumber\\
&\supset[(S^{d_x+2},s_0),(L^{d_k-2}Y,(L^{d_k-2}Y)^{\Z_2}_0)]_{\Z_2}\nonumber\\
&\supset\cdots\nonumber\\
&\supset[(S^{d_x+d_k-1},s_0),(LY,(LY)^{\Z_2}_0)]_{\Z_2}\nonumber\\
&\supset[(S^{d_x+d_k},s_0),(Y,Y^{\Z_2}_0)]_{\Z_2}\nonumber\\
&=[S^{d_x+d_k},Y]_{\Z_2}
\end{align}
By the same argument as in~\eqref{nodefect}, the result generalizes to disconnected $Y^{\Z_2}$ by repeating the above for base points in all different components. This completes the proof of theorem \ref{theorem1}.

\section{Appendix C: Stacked skyrmions}

In this part of the appendix, we give the mathematical reasons for the following two aspects of the non-stable regime:
\begin{itemize}
\item Strong invariants may break down in the presence of weak ones (in the sense of definition (i))
\item Strong topological insulators (def. (i)) can sometimes be constructed by stacking lower-dimensional insulators
\end{itemize}
We will use the example introduced in eq.~\eqref{gsnematics} which exhibits both of the above features. Its topological phases are the set $[T^2,\text{Gr}_1(\R^3)]$, which will be determined in the following, leading to the result in eq.~\eqref{nematics}. We first outline a procedure to compute $[T^2,Y]$ in general and then specialize to $Y=\text{Gr}_1(\R^3)$. 

Denoting by $(LY)_n$ the $n$-th connected component of the free loop space $LY$, the set $[T^2,Y]$ is a disjoint union of subsets labelled $(n_1,n_2)$, which contain classes whose representatives restrict to $(LY)_{n_1}$ on $S^1\times\{s_0\}$ and to $(LY)_{n_2}$ on $\{s_0\}\times S^1$. Notice that the number of elements in a sector $(n_1,n_2)$ is the same as in $(n_2,n_1)$.

The number of elements in a subset $(n_1,n_2)$ can be determined by computing $[S^1,(LY)_{n_1}]$ and counting the elements that map to $(LY)_{n_2}$ under the map $p_*$ induced by the evaluation map. For our concrete example $Y=\text{Gr}_1(\R^3)$, the free loop space $LY$ has two connected components which we denote $(LY)_0$ (containing the constant map) and $(LY)_1$ (containing all non-trivial loops, which are freely homotopic to elements in the non-trivial class of $\pi_1(\text{Gr}_1(\R^3))=\Z_2$).

For our example, it will turn out that $\pi_1((LY)_1)$ is abelian and therefore (cf.~eq.~\eqref{based})
\begin{align}
[S^1,(LY)_1]\simeq\pi_1((LY)_1).
\end{align}
Choosing a base point in $(LY)_1$, the long exact sequence associated to the free loop fibration contains the right hand side of the above equation and reads
\setlength{\arraycolsep}{1pt}
\begin{align*}
\begin{array}{ccccccccc}
\pi_2(Y)&\stackrel{\partial_2}{\to}&
\pi_1((\Omega Y)_1)&\stackrel{i_*}{\to}&
\pi_1((LY)_1)&\stackrel{p_*}{\to}&
\pi_1(Y)&\stackrel{\partial_1}{\to}&
\pi_0((\Omega Y)_1)\\
\parallel&&\parallel&&&&\parallel&&\parallel\\
\Z&&\Z&&&&\Z_2&&0
\end{array}
\end{align*}
This exact sequence is \textit{not} split as the one with a base point in $(LY)_0$ in~\eqref{short}. This entails the fact that $[S^1,(LY)_1]\ne[S^2,Y]\times[S^1,Y]=\N\times\Z_2$, since the first map $\partial_2$ is \textit{not} the constant map as in the split case, but rather multiplication by $-2$~\cite{bechtluft}. Indeed, exactness implies that $\pi_1((LY)_1)$ must be a group with exactly four elements, leaving only the possibilities $\Z_2\times\Z_2$ or $\Z_4$. In either case, it is an abelian group as previously claimed and therefore $[S^1,(LY)_1]$ also contains only four elements (rather than infinitely many). This explains how strong invariants can break down.

The other point, that strong topological insulators as in definition (i) can be stacked, is explained by the fact that $\pi_1((LY)_1)=\Z_4$ rather than $\Z_2\times\Z_2$. If $\psi:S^1\to\text{Gr}_1(\R^3)$ is a non-trivial topological insulator in one dimension, i.e. represents the non-trivial class in $\pi_1(\text{Gr}_1(\R^3))=\Z_2$, then the generator of $\pi_1((LY)_1)=\Z_4$ is represented by $\psi(k_1+k_2)$, where $k_1$ is the coordinate associated to $\pi_1$ and $k_2$ is the free loop coordinate. Since the group structure in $\pi_1$ is concatenation of loops, the other elements in $\Z_4$ are represented respectively by
\begin{align}
\psi(m k_1+k_2),
\end{align}
with $m=0,1,2,3$. These configurations are illustrated in Figs.~\ref{01} ($m=0$), \ref{21} ($m=2$) and \ref{31} ($m=3$). The ones with even $m$ belong to the sector $(1,0)$, while the ones with odd $m$ belong to the sector $(1,1)$. All of these maps correspond to the one-dimensional non-trivial insulator stacked along the $(-1,m)$-direction of the two-dimensional lattice $\Z^2$.

The above implies that the sectors $(1,1)$, $(1,0)$ and therefore also the sector $(0,1)$ contain two elements, all of which can be realized by stacking. Together with the result of appendix B, which states that the sector $(0,0)$ is in bijection with $[S^2,\text{Gr}_1(\R^3)]=\N$, the result~\eqref{nematics} follows. Of the sector $(0,0)$ only the constant map can be realized by stacking, giving a total of \textit{seven} (rather than the naively expected \textit{four}) stacked topological insulators.

Notice that this example additionally demonstrates that there can be no compromise between definitions (i) and (ii) of strong topological insulators. Indeed, the only topological phases which cannot be realized by stacking are the non-trivial elements in $\N=[S^2,\text{Gr}_1(\R^3)]\subset[T^2,\text{Gr}_1(\R^3)]$.


\begin{thebibliography}{10}

\bibitem{bernevig} B.~A.~Bernevig, T.~A.~Hughes, S.~C.~Zhang, \textit{Quantum Spin Hall Effect and Topological Phase Transition in HgTe Quantum Wells}, Science~\textbf{314},~1757~(2006)
\bibitem{fukane} L.~Fu, C.~L.~Kane, \textit{Topological insulators with inversion symmetry}, Phys.~Rev.~B~\textbf{76},~045302~(2007)
\bibitem{zhang} H.~Zhang, C.~X.~Liu, X.~L.~Qi, X.~Dai, Z.~Fang, S.~C.~Zhang, \textit{Topological insulators in Bi2Se3, Bi2Te3 and Sb2Te3 with a single Dirac cone on the surface}, Nat.~Phys.~\textbf{5},~438~(2009)
\bibitem{molenkamp} M.~K\"onig, S.~Wiedmann, C.~Br\"une, A.~Roth, H.~Buhmann, L.~W.~Molenkamp, X.~L.~Qi, S.~C.~Zhang, \textit{Quantum Spin Hall Insulator State in HgTe Quantum Wells}, Science~\textbf{318},~766~(2007)
\bibitem{hasan} D.~Hsieh, D.~Qian, L.~Wray, Y.~Xia, Y.~S.~Hor, R.~J.~Cava, M.~Z.~Hasan, \textit{A topological Dirac insulator in a quantum spin Hall phase} Nature~\textbf{452},~970~(2008)
\bibitem{hasan2} Y.~Xia, D.~Qian, D.~Hsieh, L.~Wray, A.~Pal, H.~Lin, A.~Bansil,
D.~Grauer, Y.~S.~Hor, R.~J.~Cava, M.~Z.~Hasan, \textit{Observation of a large-gap topological-insulator class with a single Dirac cone on the surface}, Nat.~Phys.~\textbf{5},~398~(2009)
\bibitem{kitaev} A.~Kitaev, \textit{Periodic table for topological insulators and superconductors}, Proceedings of the L. D. Landau Memorial Conference \bibitem{fukanemele} L.~Fu, C.~L.~Kane, E.~J.~Mele \textit{ Topological Insulators in Three Dimensions}, Phys.~Rev.~Lett.~\textbf{98}, 106803 (2007)
\bibitem{teokane} J.~Teo, C.~Kane, \textit{Topological Defects and Gapless Modes in Insulators and Superconductors}, Phys.~Rev.~B~\textbf{82}, 115120 (2010)
\bibitem{hughes} T.~L.~Hughes, E.~Prodan, B.~A.~Bernevig, \textit{Inversion-symmetric topological insulators}, Phys.~Rev.~B~\textbf{83}, 245132 (2011)
\bibitem{gurarie} A.~M.~Essin, V.~Gurarie, \textit{ Bulk-boundary correspondence of topological insulators from their Green's functions}, Phys.~Rev.~B \textbf{84}, 125132 (2011)
\bibitem{graf} G.~M.~Graf, M.~Porta, \textit{Bulk-Edge Correspondence for Two-Dimensional Topological Insulators}, Communications in Mathematical Physics, Volume 324, Issue 3, pp 851-895 (2013)
\bibitem{denittis1} G.~De~Nittis, K.~Gomi, \textit{Classification of "Real" Bloch-bundles: Topological Insulators of type AI}, arXiv:1402.1284 [math-ph] (2014)
\bibitem{denittis2} G.~De~Nittis, K.~Gomi, \textit{Classification of "Quaternionic" Bloch-bundles: Topological Insulators of type AII}, arXiv:1404.5804 [math-ph] (2014)
\bibitem{kz} R.~Kennedy, M.~R.~Zirnbauer, \textit{Bott Periodicity for $\mathbb{Z}_2$ symmetric ground states of gapped free-fermion systems} (submitted to the arXiv on Sept. 6, 2014)
\bibitem{hopf} J.~E.~Moore, Y.~ Ran, X.-G.~Wen, \textit{Topological Surface States in Three-Dimensional Magnetic Insulators}, Phys.~Rev.~Lett.~\textbf{101}, 186805 (2008)
\bibitem{hopfcalc} D.~Auckly, L.~Kapitanski, \textit{The Pontrjagin-Hopf invariants for Sobolev maps}, page~15, Commun.~Contemp.~Math. \textbf{12}, 121 (2010)
\bibitem{jaenich} K.~J\"anich, \textit{Topological Properties of Nematics in 3-space}, Acta~Applicandae~Mathematicae~\textbf{8}, 65-74 (1987)
\bibitem{bechtluft} S.~Bechtluft-Sachs, M.~Hien, \textit{The Global Defect Index}, Communications~in~Mathematical~Physics, Volume~\textbf{202}, Number~2, Pages~403-409 (1999)
\bibitem{chen} B.~G.~Chen, \textit{Topological Defects In Nematic And Smectic Liquid Crystals}, Dissertation, University of Pennsylvania (2012)
\bibitem{tomdieck2} T.~tom~Dieck, \emph{Equivariant stable homotopy theory}, Walter de Gruyter, pp.\ 95-107 (1987)
\bibitem{shah} J.~Shah, \textit{Equivariant Algebraic Topology}, Preprint (2010)
\bibitem{may} J.~P.~C.~Greenlees, J.~P.~May, \textit{Equivariant stable homotopy theory}, Handbook of algebraic
topology, North-Holland, Amsterdam (1995)
\bibitem{hhz} P.~Heinzner, A.~Huckleberry, M.~R.~Zirnbauer, \textit{Symmetry classes of disordered fermions}, Commun.~Math.~Phys.~\textbf{257},~725~(2005)
\bibitem{az} A.~Altland, M.~R.~Zirnbauer, \textit{Non-standard symmetry classes in mesoscopic normal--superconducting hybrid structures}, Phys.~Rev.~B~\textbf{55},~1142 (1997)
\bibitem{zirn2010} M.~R.~Zirnbauer, \textit{Symmetry Classes}, arXiv:1001.0722v1~(2010)
\bibitem{ludwig} A.~P.~Schnyder, S.~Ryu, A.~Furusaki, A.~Ludwig, \textit{Classification of topological insulators and superconductors in three spatial dimensions}, Phys.~Rev.~B~\textbf{78},~195125~(2008)
\bibitem{ludwig2} S.~Ryu, A.~P.~Schnyder, A.~Furusaki, A.~Ludwig, \textit{Topological insulators and superconductors: tenfold way and dimensional hierarchy}, New~J.~Phys.~\textbf{12}~065010 (2010)
\bibitem{bott} R.~Bott, \textit{The stable homotopy of the classical groups}, Ann.~of~Math.~\textbf{70},~313-337~(1959)
\bibitem{milnor} J.~Milnor, \textit{Morse Theory}, Princeton~University~Press~(1963)
\bibitem{stone} M.~Stone, C.~Chiu, A.~Roy, \textit{Symmetries, dimensions and topological insulators: the mechanism behind the face of the Bott clock}, J. Phys. A: Math. Theor. \textbf{44} 045001 (2011)
\bibitem{quast} A.~Mare, P.~Quast, \textit{Bott Periodicity of Inclusions}, arXiv:1108.0954v1~(2011)
\bibitem{tomdieck} T.~tom~Dieck, \textit{Algebraic Topology}, European~Mathematical~Society~(2008)
\bibitem{whiteheadbook} G.~W.~Whitehead \textit{Elements of Homotopy Theory}, Springer-Verlag New York (1978)
\bibitem{hatcher} A.~Hatcher, \textit{Algebraic Topology}, Cambridge~University~Press~(2002)
\bibitem{inversion} A.~M.~Turner, Y.~Zhang, R.~S.~K.~Mong, A.~Vishwanath, \textit{Quantized Response and Topology of Insulators with Inversion Symmetry}, arXiv:1010.4335v2~(2010)
\bibitem{stern} Z.~Ringel, Y.~E.~Kraus, A.~Stern, \textit{The strong side of weak topological insulators}, arXiv:1105.4351v1 (2011)

\end{thebibliography}
\end{document}